\documentclass{article}
\usepackage{graphicx} 
\usepackage{amsmath}
\usepackage{amssymb}
\usepackage{mathtools}
\usepackage{amsthm}
\usepackage{tikz}
\usetikzlibrary{positioning}
\usetikzlibrary{shapes.geometric}
\usetikzlibrary{shapes.misc}
\usepackage{hyperref}

\title{Battle Sheep is PSPACE-complete}
\author{Kyle Burke, Hirotaka Ono }

\newcommand{\rsBSheep}{\ruleset{Battle Sheep}}
\newcommand{\bcl}{\ruleset{B2CL}}
\newcommand{\ruleset}[1]{\textsc{#1}}

\newcommand{\cclass}[1]{\ensuremath{\mathord{\textrm{#1}}}} 

\newtheorem{theorem}{Theorem}
\newtheorem{open}{Open Problem}
\newtheorem{corollary}{Corollary}

\newcommand{\sheepHex}[4]{\node[regular polygon, regular polygon sides=6, minimum width=2cm, draw] (#1) at #4 {};
\node[circle, draw=#2, text=#2, minimum width=1cm, draw] (Stack-#1) at #4 {\large{#3}};}

\newcommand{\emptyHex}[3]{\node[regular polygon, regular polygon sides=6, minimum width=2cm, draw] (#1) at #3 {#2};}

\newcommand{\blockHex}[2]{\node[regular polygon, regular polygon sides=6, minimum width=2cm, draw] (#1) at #2 {};
\node[circle, minimum width=1cm, draw] (Block-#1) at #2 {\large{1}};}

\begin{document}

\maketitle

\begin{abstract}
    \rsBSheep{} is a board game published by Blue Orange Games.  With two players, it is a combinatorial game that uses normal play rules.  We show that it is \cclass{PSPACE}-complete, even when each stack has only up to 3 tokens.
\end{abstract}

\section{Introduction}

\subsection{Game Origin and Ruleset}

Battle Sheep\footnote{\url{https://www.blueorangegames.com/games/battle-sheep} .  It was originally given the name \emph{Splits}.} is a board game published by Blue Orange Games\footnote{\url{https://www.blueorangegames.com/}} at the time of this writing.  The game consists of stacks of sheep-styled tokens of different colors, with each color belonging to a player (Black, White, Blue, and Red).  The stacks are arranged on a board of hexagons that forms a part of a triangular grid.  The sum total of all tokens in stacks is equal to the number of spaces on the board.  A legend of our diagram notation is included in Figure \ref{fig:legend}. 

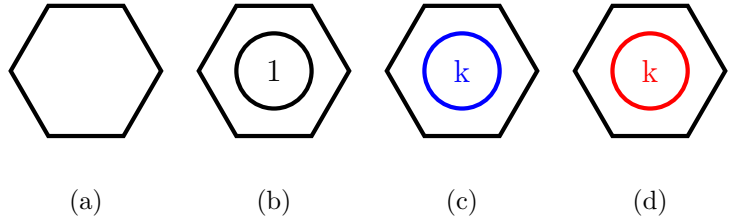
\begin{figure}[h!]
    \begin{center}
    \begin{tikzpicture}[node distance = .75cm, minimum size = .5cm, inner sep = .07cm, ultra thick]
        \emptyHex{blankA}{}{(0, 0)}
        \node[] (lbl1) at (0, -1.75) {(a)};
        \blockHex{blockA}{(2.5, 0)}
        \node[] (lbl2) at (2.5, -1.75) {(b)};
        \sheepHex{blue}{blue}{k}{(5,0)}
        \node[] (lbl3) at (5, -1.75) {(c)};
        \sheepHex{red}{red}{k}{(7.5,0)}
        \node[] (lbl4) at (7.5, -1.75) {(d)};
    \end{tikzpicture}
    \end{center}
    \caption{Legend: (a) is a blank space with no sheep tokens.  (b) is a space with 1 token.  Since the color doesn't matter, as there are no excess sheep to move, we leave them black.  (c) is a space with $k$ Blue sheep tokens.  (d) is a space with $k$ Red sheep tokens.}
    \label{fig:legend}
\end{figure} 

On their turn, a player chooses a stack of their color with at least two tokens, picks some sub-stack of the tokens to move (the bottom sheep token cannot move) and a direction towards a neighboring space.  The player then moves those tokens as far as they have a clear path in that direction, stopping either before running off the board or hitting another stack.  A player must go all the way; they can't choose to stop part way along that path.  Examples of allowed moves are shown in figure \ref{fig:moves}.

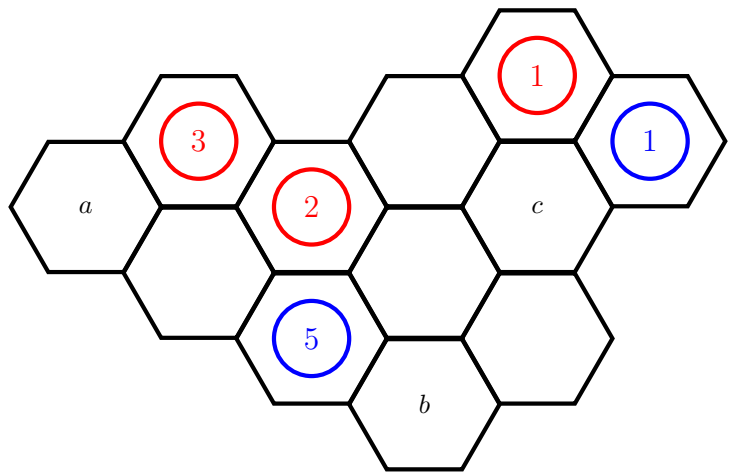
\begin{figure}[h!]
    \begin{center}
    \begin{tikzpicture}[node distance = .75cm, minimum size = .5cm, inner sep = .07cm, ultra thick]
        \emptyHex{center}{}{(0, 0)}
        \emptyHex{z}{}{(0, 1.75)}
        \sheepHex{red2}{red}{2}{(-1.5,.875)}
        \sheepHex{red3}{red}{3}{(-3.0, 1.75)}
        \emptyHex{a}{$a$}{(-4.5, .875)}
        \emptyHex{y}{}{(-3, 0)}
        \sheepHex{blue5}{blue}{5}{(-1.5,-.875)}
        \emptyHex{b}{$b$}{(0, -1.75)}
        \emptyHex{x}{}{(1.5, -.875)}
        \emptyHex{c}{$c$}{(1.5, .875)}
        \sheepHex{blue1}{blue}{1}{(3.0, 1.75)}
        \sheepHex{red1}{red}{1}{(1.5, 2.625)}
    \end{tikzpicture}
    \end{center}
    \caption{A position in Battle Sheep.  Blue's options are to move 1 to 4 tokens from the pile with 5 to space $a$, $b$, or $c$.  Blue cannot move straight up because there are tokens directly in the way.  They cannot move to the space between their stack and $a$ because they must move as far as possible in the chosen direction.  The same is true of the space between the 5-stack and $c$.  They cannot move from the stack with 1 sheep token because they must leave at least one token behind and there are not any excess tokens to move.}
    \label{fig:moves}
\end{figure} 

If no possible moves are available, either because all their stacks only have one token or all their stacks with more than one sheep are surrounded by other tokens and/or edges of the board, then that player loses.  

When played with two players, the rules function exactly as a combinatorial game under Normal Play:
\begin{itemize}
    \item Two players take turns making moves.
    \item There is no hidden information.
    \item There are no random elements.
    \item Normal Play: when a player can't move on their turn, they lose the game. \cite{WinningWays:2001}, \cite{LessonsInPlay:2007}, \cite{SiegelCGT:2013}
\end{itemize}
Because no changes need to be made, we use the same name for the ruleset: \rsBSheep{}.  We provide a playable version of \rsBSheep{} at \url{https://kyleburke.info/DB/combGames/battleSheep.html}.  

\subsection{Computational Complexity}

Computational hardness is used to give evidence that some problems are intractible: there may not be a polynomial-time algorithm to solve them.  This intractibility is demonstrated by showing that a problem can't be solved in polynomial time unless the most difficult problems for an unsolved computational complexity class can also be solved in polynomial time.  This classification is performed by finding a reduction from a known difficult problem to the problem in question.  A reduction is a transformation that preserves the result of the problem that can be performed in polynomial time \cite{Papadimitriou:1994}.  For example, in this paper, to prove that \rsBSheep{} is \cclass{PSPACE}-hard, we find a reduction, $f: \bcl{} \rightarrow \rsBSheep{}$ such that the first player wins \bcl{} position $x$ exactly when the first player wins \rsBSheep{} position $f(x)$.  

\section{\cclass{PSPACE}-completeness}

\begin{theorem}[main]
    \rsBSheep{} is \cclass{PSPACE}-hard when the total number of tokens is equal to the total number of spaces on the board.
\end{theorem}

\begin{proof}
    We prove this by reducing from \ruleset{Bounded Two-Player Constraint Logic} (\bcl{}).  To complete this reduction, we only need to demonstrate a working construction for the Variable, Goal, Split, Choice, And, and Or gadgets, as this includes \cclass{PSPACE}-hard instances of \bcl{}  \cite{DBLP:books/daglib/0023750}.  Without loss of generality, we will assume the next player to move is Blue.  Our reduction works if Blue can win the resulting \rsBSheep{} position exactly when they can win the \bcl{} position.  In \rsBSheep{}, Blue will win when they can
    \begin{itemize}
        \item Move in such a way that satisfies the boolean-logic nature of all gadgets, and
        \item Move on the final Goal gadget.
    \end{itemize}

    The \bcl{} positions we reduce from model positive circuits, meaning there are no negations.  It will always be in Blue's interest to activate the output of any gadget, as long as the gate logic of the gadget holds.  We will enforce this last point by providing Red with independent extra moves so that Blue will need to make a specific number of plays on each gadget.  If they do not, then Red will have additional moves to make after Blue has made all their moves, including the winning move on the Goal gadget.  This is a common construction strategy for combinatorial game hardness reductions.  \cite{BurkeHearn2018}

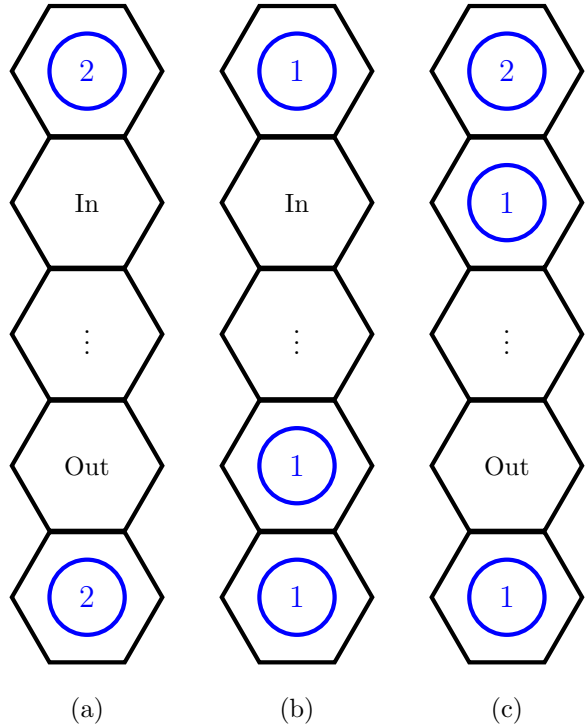
\begin{figure}[h!]
    \begin{center}
    \begin{tikzpicture}[node distance = .75cm, minimum size = .5cm, inner sep = .07cm, ultra thick]
        \sheepHex{blueIn}{blue}{2}{(0,0)}
        \emptyHex{EmptyIn}{In}{(0, -1.75)}
        \emptyHex{ellipse}{$\vdots$}{(0, -3.5)}
        \emptyHex{EmptyOut}{Out}{(0, -5.25)}
        \sheepHex{blueOut}{blue}{2}{(0, -7)}
        \node[] (lbl) at (0, -8.5) {(a)};
    \end{tikzpicture}
    \hspace{.5cm}
    \begin{tikzpicture}[node distance = .75cm, minimum size = .5cm, inner sep = .07cm, ultra thick]
        \sheepHex{blueIn}{blue}{1}{(0,0)}
        \emptyHex{EmptyIn}{In}{(0, -1.75)}
        \emptyHex{ellipse}{$\vdots$}{(0, -3.5)}
        \sheepHex{bluePlayed}{blue}{1}{(0, -5.25)}
        \sheepHex{blueOut}{blue}{1}{(0, -7)}
        \node[] (lbl) at (0, -8.5) {(b)};
    \end{tikzpicture}
    \hspace{.5cm}
    \begin{tikzpicture}[node distance = .75cm, minimum size = .5cm, inner sep = .07cm, ultra thick]
        \sheepHex{blueIn}{blue}{2}{(0,0)}
        \sheepHex{bluePlayed}{blue}{1}{(0, -1.75)}
        \emptyHex{ellipse}{$\vdots$}{(0, -3.5)}
        \emptyHex{EmptyOut}{Out}{(0, -5.25)}
        \sheepHex{blueOut}{blue}{1}{(0, -7)}
        \node[] (lbl) at (0, -8.5) {(c)};
    \end{tikzpicture}
    \end{center}
    \caption{A Wire between two gadgets.  The output of the bottom gadget is connected to the input of the top gadget.  The ellipsis denotes that there could be many hexagons along that path.  (a) is how the Wire begins as a connection between two gadgets.  (b) is the result of the case where Blue was able to play the token below the output to another location, meaning the extra token on the input side could come down and block it.  This is the active case, where a positive signal is sent from the lower gadget to the upper.  (c) is the result of the case where Blue was not able to play a token onto the lower gadget, and instead had to send it up, propagating an inactive signal.  Blue will have to find another place to play the top token of the upper stack.}
    \label{fig:connection}
\end{figure}

    To model this, we model the active signal with a space left empty on the output side, leaving Blue a space to move to as the input for the gadget on the other side.  Each gadget will have input or output hexes (most have both).  The connections between them will work as shown in figure \ref{fig:connection}.  The active signal (or "on" or "true") is passed when the token next to the output is able to move elsewhere aside from towards the input.  This allows the extra sheep token adjacent to the output to move down to the output space, freeing up space in the top gadget.

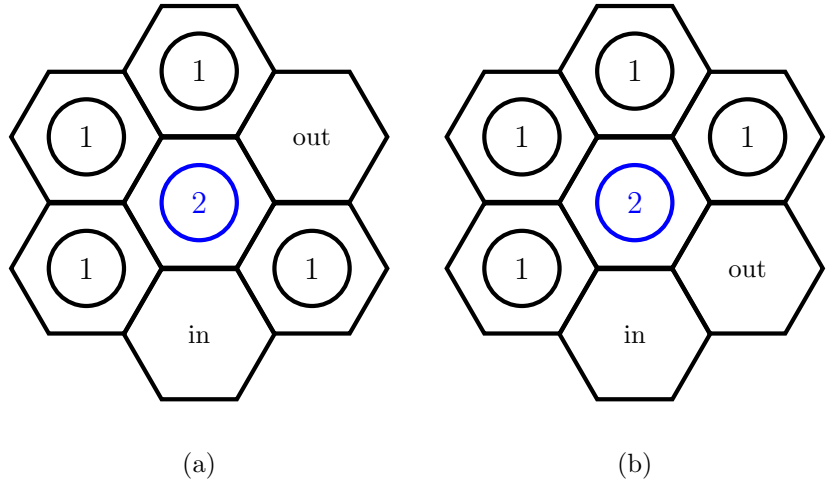
\begin{figure}[h!]
    \begin{center}
    \begin{tikzpicture}[node distance = .75cm, minimum size = .5cm, inner sep = .07cm, ultra thick]
        \sheepHex{blue}{blue}{2}{(0,0)}
        \emptyHex{in}{in}{(0, -1.75)}
        \emptyHex{out}{out}{(1.5, .875)}
        \blockHex{blockA}{(-1.5, -.875)}
        \blockHex{blockB}{(-1.5, .875)}
        \blockHex{blockC}{(0, 1.75)}
        \blockHex{blockD}{(1.5, -.875)}
        \node[] (lbl) at (0, -3.5) {(a)};
    \end{tikzpicture}
    \hspace{.5cm}
    \begin{tikzpicture}[node distance = .75cm, minimum size = .5cm, inner sep = .07cm, ultra thick]
        \sheepHex{blue}{blue}{2}{(0,0)}
        \emptyHex{in}{in}{(0, -1.75)}
        \emptyHex{out}{out}{(1.5, -.875)}
        \blockHex{blockA}{(-1.5, -.875)}
        \blockHex{blockB}{(-1.5, .875)}
        \blockHex{blockC}{(0, 1.75)}
        \blockHex{blockD}{(1.5, .875)}
        \node[] (lbl) at (0, -3.5) {(b)};
    \end{tikzpicture}
    \end{center}
    \caption{Wire turns of 30 and 60 degrees.  In both cases, an active input allows the output to be activated.}
    \label{fig:turns}
\end{figure}

    In figure \ref{fig:turns}, we show how we can adjust the directions in wires with a turning gadget.  Turns of 30 and 60 degrees in either direction allow full flexibility to arrange the gadgets with each other.

\begin{figure}[h!]
    \begin{center}
    \begin{tikzpicture}[node distance = .75cm, minimum size = .5cm, inner sep = .07cm, ultra thick]
        \node[regular polygon, regular polygon sides=6, minimum width=2cm,draw] (top) at (0,0){};
        \node[circle, draw=blue, text=blue, minimum width=1cm, draw] (topStack) at (0,0){\large{2}};
        
        \node[regular polygon, regular polygon sides=6, minimum width=2cm,draw] (border1) at (0, 1.75){};
        \node[circle, minimum width=1cm, draw] (border1Stack) at (0,1.75){\large{1}};
        
        \node[regular polygon, regular polygon sides=6, minimum width=2cm,draw] (border2) at (1.5, .875){};
        \node[circle, minimum width=1cm, draw] (border2Stack) at (1.5, .875){\large{1}};
        
        \node[regular polygon, regular polygon sides=6, minimum width=2cm,draw] (border3) at (1.5, -.875){};
        \node[circle, minimum width=1cm, draw] (border3Stack) at (1.5, -.875){\large{1}};
        
        \node[regular polygon, regular polygon sides=6, minimum width=2cm,draw] (border4) at (-1.5, -.875){};
        \node[circle, minimum width=1cm, draw] (border4Stack) at (-1.5, -.875){\large{1}};
        
        \node[regular polygon, regular polygon sides=6, minimum width=2cm,draw] (border5) at (-1.5, .875){};
        \node[circle, minimum width=1cm, draw] (border5Stack) at (-1.5, .875){\large{1}};
        
        \node[regular polygon, regular polygon sides=6, minimum width=2cm,draw] (reg1) at (0, -1.75){In};
    \end{tikzpicture}
    \end{center}
    \caption{Goal gadget.  If the input is active (empty), then the Blue player is able to move one token there to win the game.  The black sheep stacks can belong to either player since there is only one sheep in each stack.}
    \label{fig:goal}
\end{figure}
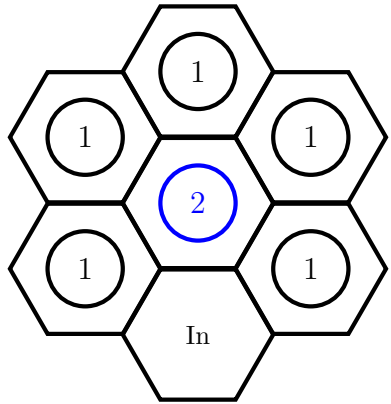
    
    Of the gadgets needed to fulfill the \bcl{} requirements, he most simple is the Goal gadget, shown in figure \ref{fig:goal}.  If the input to the Goal is active (empty) after all other gadgets have been completed correctly, then Red should have made all moves available to them and Blue wins by moving a sheep to the space labeled In.  If not, then Blue will be unable to move the extra token in the goal gadget and will lose the game.

\begin{figure}[h!]
    \begin{center}
    \begin{tikzpicture}[node distance = .75cm, minimum size = .5cm, inner sep = .07cm, ultra thick]
        \sheepHex{red}{red}{2}{(0,0)}
        \emptyHex{EmptyA}{}{(0, -1.75)}
        \emptyHex{EmptyB}{}{(1.5, -.875)}
        \sheepHex{blue}{blue}{2}{(1.5, .875)}
        \emptyHex{EmptyC}{Out}{(1.5, 2.625)}
        \blockHex{blockA}{(0, 1.75)}
        \blockHex{blockB}{(-1.5, .875)}
        \blockHex{blockC}{(-1.5, -.875)}
        \blockHex{blockD}{(-1.5, -2.625)}
        \blockHex{blockE}{(0, -3.5)}
        \blockHex{blockF}{(1.5, -2.625)}
        \blockHex{blockG}{(3.0, -1.75)}
        \blockHex{blockG}{(3.0, 0)}
        \blockHex{blockG}{(3.0, 1.75)}
    \end{tikzpicture}
    \end{center}
    \caption{Variable gadget at its initial state.}
    \label{fig:variable}
\end{figure}
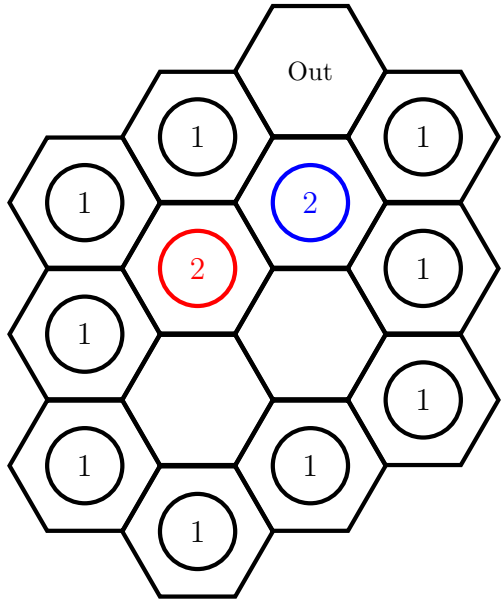

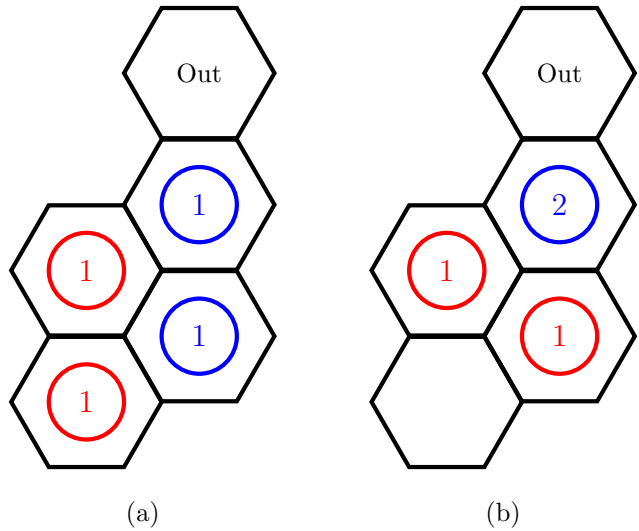
\begin{figure}[h!]
    \begin{center}
    \begin{tikzpicture}[node distance = .75cm, minimum size = .5cm, inner sep = .07cm, ultra thick]
        \sheepHex{red}{red}{1}{(0,0)}
        \sheepHex{redMove}{red}{1}{(0, -1.75)}
        \sheepHex{blueMove}{blue}{1}{(1.5, -.875)}
        \sheepHex{blue}{blue}{1}{(1.5, .875)}
        \emptyHex{EmptyC}{Out}{(1.5, 2.625)}
        \node[] (lbl) at (.75, -3.25) {(a)};
    \end{tikzpicture}
    \hspace{1cm}
    \begin{tikzpicture}[node distance = .75cm, minimum size = .5cm, inner sep = .07cm, ultra thick]
        \sheepHex{red}{red}{1}{(0,0)}
        \emptyHex{Empty}{}{(0, -1.75)}
        \sheepHex{redMove}{red}{1}{(1.5, -.875)}
        \sheepHex{blue}{blue}{2}{(1.5, .875)}
        \emptyHex{EmptyC}{Out}{(1.5, 2.625)}
        \node[] (lbl) at (.75, -3.25) {(b)};
    \end{tikzpicture}
    \end{center}
    \caption{Variable gadget after moves.  (a) is the result of both moves if Blue moves first, which activates the output.  (b) is Red's best first move.  In this case, Blue will have to move their excess token out of the gadget, leaving it inactive.}
    \label{fig:variable2}
\end{figure}

    Next we describe the Variable gadget, shown in figures \ref{fig:variable} and \ref{fig:variable2}.  In this gadget, both players want to be the first to play.  If Blue goes first, they can move a sheep down, leaving the output clear so the variable is True (active).  If Red goes first, they set the variable to a False value by moving a token to the right, forcing Blue to move up and deactivate the output.

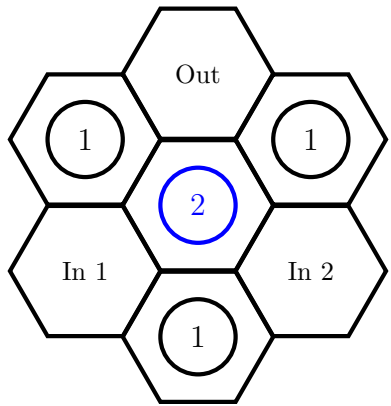
\begin{figure}[h!]
    \begin{center}
    \begin{tikzpicture}[node distance = .75cm, minimum size = .5cm, inner sep = .07cm, ultra thick]
        \sheepHex{blue}{blue}{2}{(0,0)}
        \emptyHex{EmptyA}{In 1}{(-1.5, -.875)}
        \blockHex{blockA}{(-1.5, .875)}
        \emptyHex{EmptyB}{In 2}{(1.5, -.875)}
        \blockHex{blockB}{(1.5, .875)}
        \blockHex{blockE}{(0, -1.75)}
        \emptyHex{EmptyC}{Out}{(0, 1.75)}
    \end{tikzpicture}
    \end{center}
    \caption{Or gadget.  Blue can keep the output clear if either of the inputs are clear. }
    \label{fig:or}
\end{figure}

    In the Or gadget, shown in figure \ref{fig:or}, Blue can activate the output whenever either of the inputs are active.

\begin{figure}[h!]
    \begin{center}
    \begin{tikzpicture}[node distance = .75cm, minimum size = .5cm, inner sep = .07cm, ultra thick]
        \sheepHex{blueMiddle}{blue}{3}{(0,0)}
        \blockHex{blockA}{(-3, -3.5)}
        \blockHex{blockB}{(-1.5, -2.625)}
        \blockHex{blockC}{(0, -1.75)}
        \blockHex{blockD}{(1.5, -2.625)}
        \blockHex{blockE}{(3, -3.5)}

        \emptyHex{middleEmpty}{}{(0, 1.75)}
        \emptyHex{belowEmpty}{}{(0, 3.5)}
        \sheepHex{blueIn}{blue}{2}{(0, 5.25)}
        \emptyHex{out}{Out}{(0, 7)}
        
        \emptyHex{leftEmpty}{}{(-1.5, -.875)}
        \sheepHex{blueLeft}{blue}{2}{(-3, -1.75)}
        \emptyHex{inLeft}{In 1}{(-4.5, -2.625)}
        \blockHex{blockF}{(-4.5, -.875)}
        \blockHex{blockG}{(-3, 0)}
        \blockHex{blockH}{(-1.5, .875)}
        \blockHex{blockI}{(-1.5, 2.625)}
        \blockHex{blockJ}{(-1.5, 4.375)}
        \blockHex{blockK}{(-1.5, 6.125)}
        
        \emptyHex{rightEmpty}{}{(1.5, -.875)}
        \sheepHex{blueRight}{blue}{2}{(3, -1.75)}
        \emptyHex{inRight}{In 2}{(4.5, -2.625)}
        \blockHex{blockF2}{(4.5, -.875)}
        \blockHex{blockG2}{(3, 0)}
        \blockHex{blockH2}{(1.5, .875)}
        \blockHex{blockI2}{(1.5, 2.625)}
        \blockHex{blockJ2}{(1.5, 4.375)}
        \blockHex{blockK2}{(1.5, 6.125)}
    \end{tikzpicture}
    \end{center}
    \caption{And Gadget.  If both inputs are active, then the stack of three in the middle can make moves to both the left and right, freeing up the top stack to send a token down and activate.  If either of the inputs are inactive, then at least one (and maybe both) of the tokens in the middle must move up, keeping the output inactive.}
    \label{fig:and}
\end{figure}
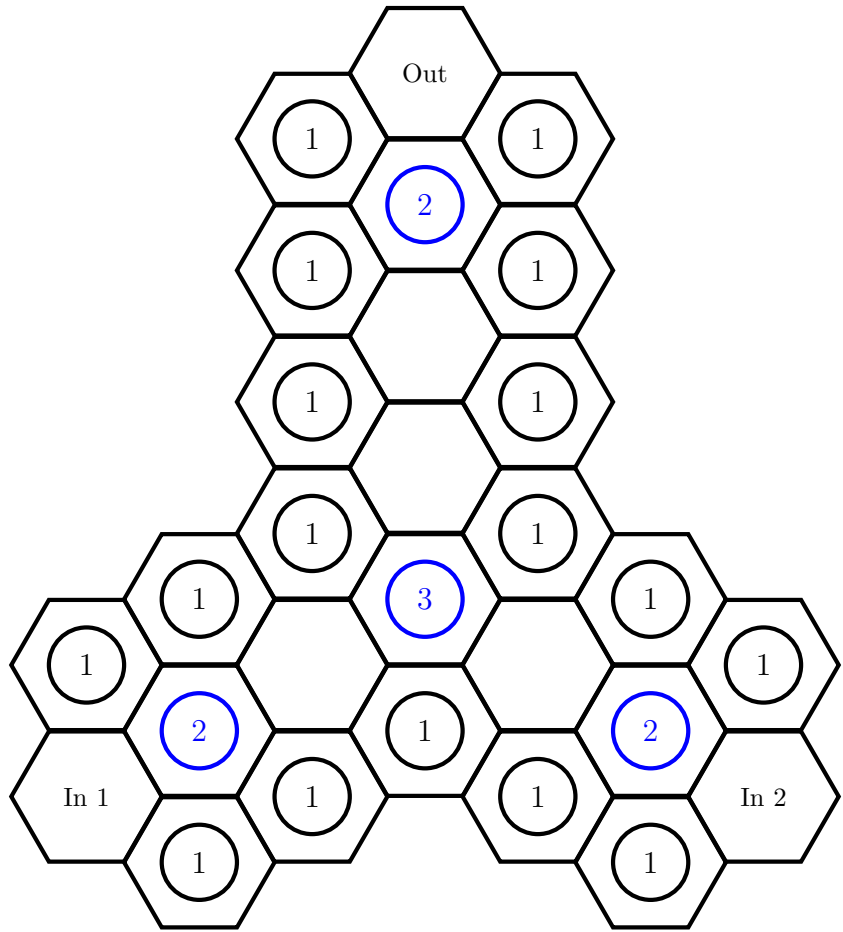

    For the And gadget, shown in figure \ref{fig:and}, the output will be active only if both inputs are active.  In that case, the two stacks on the sides can move to the inputs, leaving the two spaces below the stack of three blue sheep empty.

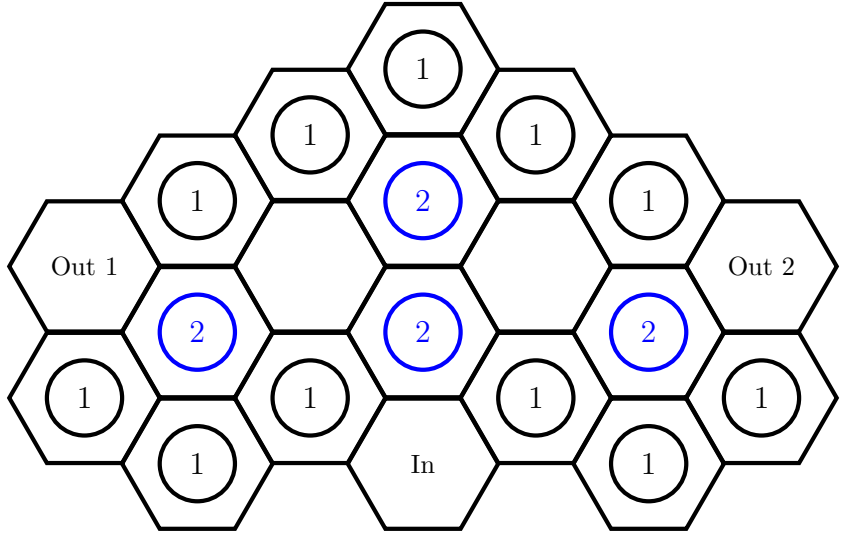
\begin{figure}[h!]
    \begin{center}
    \begin{tikzpicture}[node distance = .75cm, minimum size = .5cm, inner sep = .07cm, ultra thick]
        \sheepHex{blue}{blue}{2}{(0,0)}
        \sheepHex{blueTop}{blue}{2}{(0, 1.75)}
        \emptyHex{EmptyA}{}{(-1.5, .875)}
        \blockHex{blockA}{(-1.5, -.875)}
        \sheepHex{blueLeft}{blue}{2}{(-3, 0)}
        \emptyHex{inA}{Out 1}{(-4.5, .875)}
        \emptyHex{EmptyB}{}{(1.5, .875)}
        \blockHex{blockB}{(1.5, -.875)}
        \sheepHex{blueRight}{blue}{2}{(3, 0)}
        \emptyHex{EmptyC}{In}{(0, -1.75)}
        \emptyHex{inB}{Out 2}{(4.5, .875)}
        \blockHex{blockAA}{(-3, -1.75)}
        \blockHex{blockAB}{(-4.5, -.875)}
        \blockHex{blockC}{(-3, 1.75)}
        \blockHex{blockD}{(-1.5, 2.625)}
        \blockHex{blockE}{(0, 3.5)}
        \blockHex{blockF}{(1.5, 2.625)}
        \blockHex{blockG}{(3, 1.75)}
        \blockHex{blockH}{(4.5, -.875)}
        \blockHex{blockI}{(3, -1.75)}
    \end{tikzpicture}
    \end{center}
    \caption{Choice Gadget.  If the input is active, then Blue can choose one of the outputs to be active by moving the excess token to the opposite side.  If the input is inactive, then both outputs must be inactive.}
    \label{fig:choice}
\end{figure}

    The Choice gadget, shown in figure \ref{fig:choice}, is an inverted version of And.  If the input is inactive, then both outputs will be made inactive, because the two middle piles of 2 will need to send tokens to the empty spaces in the interior.  If the input is active, then Blue can choose one of the outputs to be active by moving the top extra token to the opposite side they want to activate.

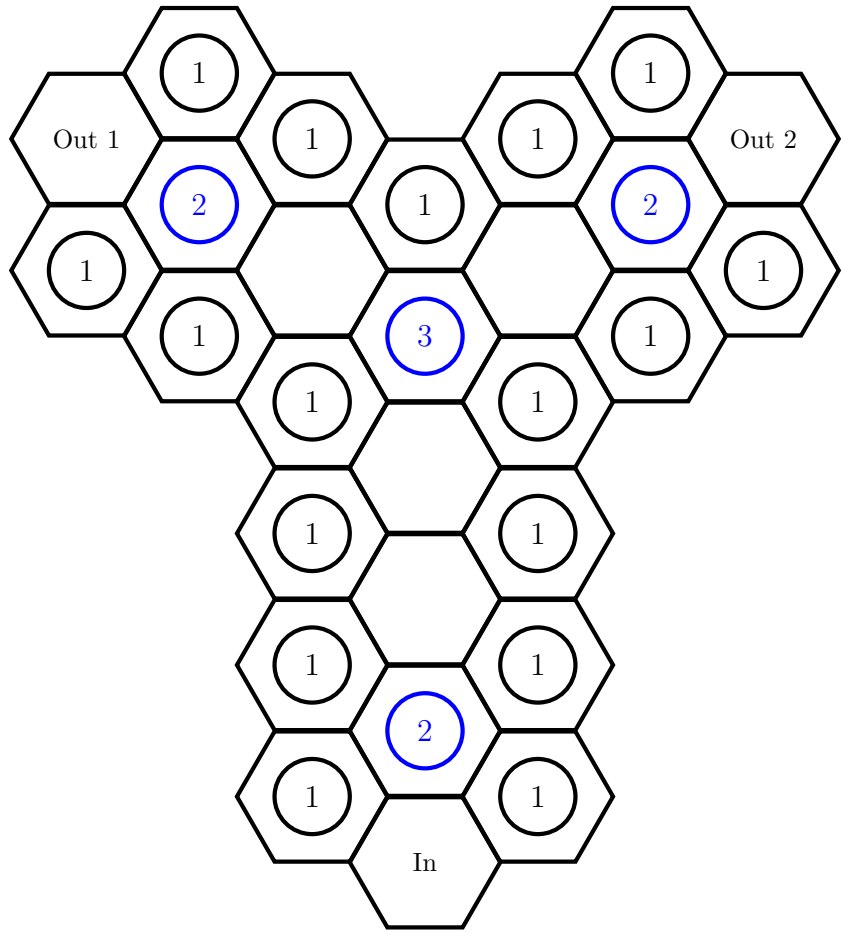
\begin{figure}[h!]
    \begin{center}
    \begin{tikzpicture}[node distance = .75cm, minimum size = .5cm, inner sep = .07cm, ultra thick]
        \sheepHex{blueMiddle}{blue}{3}{(0,0)}
        \blockHex{blockA}{(-3, 3.5)}
        \blockHex{blockB}{(-1.5, 2.625)}
        \blockHex{blockC}{(0, 1.75)}
        \blockHex{blockD}{(1.5, 2.625)}
        \blockHex{blockE}{(3, 3.5)}

        \emptyHex{middleEmpty}{}{(0,-1.75)}
        \emptyHex{belowEmpty}{}{(0, -3.5)}
        \sheepHex{blueIn}{blue}{2}{(0, -5.25)}
        \emptyHex{in}{In}{(0, -7)}
        
        \emptyHex{leftEmpty}{}{(-1.5, .875)}
        \sheepHex{blueLeft}{blue}{2}{(-3, 1.75)}
        \emptyHex{outLeft}{Out 1}{(-4.5, 2.625)}
        \blockHex{blockF}{(-4.5, .875)}
        \blockHex{blockG}{(-3, 0)}
        \blockHex{blockH}{(-1.5, -.875)}
        \blockHex{blockI}{(-1.5, -2.625)}
        \blockHex{blockJ}{(-1.5, -4.375)}
        \blockHex{blockK}{(-1.5, -6.125)}
        
        \emptyHex{rightEmpty}{}{(1.5, .875)}
        \sheepHex{blueRight}{blue}{2}{(3, 1.75)}
        \emptyHex{outRight}{Out 2}{(4.5, 2.625)}
        \blockHex{blockF2}{(4.5, .875)}
        \blockHex{blockG2}{(3, 0)}
        \blockHex{blockH2}{(1.5, -.875)}
        \blockHex{blockI2}{(1.5, -2.625)}
        \blockHex{blockJ2}{(1.5, -4.375)}
        \blockHex{blockK2}{(1.5, -6.125)}
    \end{tikzpicture}
    \end{center}
    \caption{Fanout Gadget.  If the input is active, then Blue can activate both outputs by moving the two tokens from the stack of three down.  Otherwise, they will be forced to leave both outputs inactive.}
    \label{fig:fanout}
\end{figure}

    The Fanout gadget is shown in figure \ref{fig:fanout}.  In this construction, if the input is active, then both outputs must be able to be activated, otherwise neither of them should be.  The central stack of three blue tokens is the mechanism that achieves this.  
    
    If the input is active, then the two excess central tokens can move down, which allows both outputs to be activated.  If, instead, the input is inactive, then the two excess tokens cannot go down because they would block the extra token adjacent to the input.  Instead, they must use the two moves to fill in the upper empty spaces, making the two outputs inactive.

    In order to complete the construction, we require another "Makeup" gadget to allow a place for Red to make moves after the Variables have been claimed.  In the overall construction:
    \begin{itemize}
        \item Blue moves first,
        \item Aside from the Goal gadget, Blue will be able to move all of their excess tokens, and
        \item If Blue can resolve the \bcl{} formula to True (if they can activate the Goal gadget) then they will be able to move the excess token off of the Goal as well.
    \end{itemize}

    That last move should win the game for Blue if it's possible, but they should lose otherwise.  Because of that, we want Red to make exactly the same number of moves across all gadgets not including the Goal.  Then, if Blue cannot activate the Goal, Red makes the last move and wins.  If Blue can activate the Goal, then they move one more time and win.

    Since Blue and Red have the same number of moves to make on Variable gadgets, we only need the Makeup gadget to account for Blue's moves on Wire Turns, Or, And, Choice, and Fanout.  For each of these, Blue moves, respectively, 1, 1, 4, 4, and 5 times.  Thus, the total number of Red moves on the Makeup gadget should be $k = a + b + 4c + 4d + 5e$, where $a$ is the number of Wire Turns, $b$ is the number of Or gadgets, $c$ is the number of Ands, $d$ is the number of Choices, and $e$ is the number of Fanouts.

    The Makeup gadget, then, consists of a linear stretch of $k+1$ hexagons surrounded by blocked spaces with a single stack of $k+1$ red sheep tokens on one space.  This provides Red the needed turns to take while Blue attempts to properly activate the Goal gadget. 

    With these listed gadgets, we can assemble the \bcl{} circuit in \rsBSheep, completing the reduction.  
\end{proof}

We can use aspects of our proof to show that the resulting game is not just \cclass{PSPACE}-hard, but also \cclass{PSPACE}-complete.

\begin{corollary}[Completeness]
    \rsBSheep{} is \cclass{PSPACE}-complete.
\end{corollary}

\begin{proof}
    In our gadgets, the total number of tokens on the board is less than or equal to the number of open spaces.  This means that the depth of the game tree is polynomial in the size of the board.  (We leave the case where there are significantly more tokens to open problem \ref{open:exponentialTokens}.)  Thus we can traverse the entire game tree using a polynomial amount of bits and the problem is solvable in \cclass{PSPACE}.
\end{proof}

\section{Future Work and Open Problems}

This paper shows that \rsBSheep{} is \cclass{PSPACE}-hard.  In the cases where the largest pile size is polynomial in the number of spaces, the positions are in \cclass{PSPACE} as well, so the game is \cclass{PSPACE}-complete.  In the published game, the number of spaces is always equal to the total number of tokens.  Families of positions with exponential numbers of tokens, while being solvable in \cclass{EXPTIME}, might not be in \cclass{PSPACE}.

\begin{open}
    \label{open:exponentialTokens}
    Is \rsBSheep{} \cclass{EXPTIME}-hard when there are stacks of tokens exponential in the number of spaces on the board?
\end{open}

On the other end, our constructions with piles of sizes 1, 2, and 3.  Although we can modify the Makeup gadget to use $k$ piles of size 2 instead of a single $k+1$ pile, it's not clear whether we can modify the And and Fanout gadgets.  

\begin{open}
    Is \rsBSheep{} \cclass{PSPACE}-hard when the stacks have only one or two tokens?
\end{open}

\bibliographystyle{plainurl}

\begin{thebibliography}{1}

\bibitem{LessonsInPlay:2007}
M.~H. Albert, R.~J. Nowakowski, and D.~Wolfe.
\newblock {\em Lessons in Play: An Introduction to Combinatorial Game Theory}.
\newblock A. K. Peters, Wellesley, Massachusetts, 2007.

\bibitem{WinningWays:2001}
Elwyn~R. Berlekamp, John~H. Conway, and Richard~K. Guy.
\newblock {\em Winning Ways for your Mathematical Plays}, volume~1.
\newblock A K Peters, Wellesley, Massachsetts, 2001.

\bibitem{BurkeHearn2018}
Kyle Burke and Robert~A. Hearn.
\newblock Pspace-complete two-color planar placement games.
\newblock {\em International Journal of Game Theory}, Jul 2018.
\newblock \href {https://doi.org/10.1007/s00182-018-0628-8}
  {\path{doi:10.1007/s00182-018-0628-8}}.

\bibitem{DBLP:books/daglib/0023750}
Robert~A. Hearn and Erik~D. Demaine.
\newblock {\em Games, puzzles and computation}.
\newblock A {K} Peters, 2009.

\bibitem{Papadimitriou:1994}
C.H. Papadimitriou.
\newblock {On the complexity of the parity argument and other inefficient
  proofs of existence}.
\newblock {\em Journal of Computer and System Sciences}, pages 498--532, 1994.

\bibitem{SiegelCGT:2013}
A.N. Siegel.
\newblock {\em Combinatorial Game Theory}.
\newblock Graduate Studies in Mathematics. American Mathematical Society, 2013.
\newblock URL: \url{https://books.google.com/books?id=VUVrAAAAQBAJ}.

\end{thebibliography}

\end{document}